\documentclass[twocolumn]{svjour3}
\usepackage{cite}
\usepackage{amsmath,amssymb,bm}
\usepackage{url}
\usepackage{multirow,bigstrut}
\usepackage{graphicx}
\graphicspath{{figures/}}

\newlength\figwidth
\journalname{Nonlinear Dynamics}

\begin{document}

\title{Breaking a chaotic image encryption algorithm based on perceptron model}

\author{Yu Zhang  \and Chengqing Li \and Qin Li \and Dan Zhang \and Shi Shu}

\authorrunning{Y. Zhang et al.}

\institute{Yu Zhang, Shi Shu\at
               School of Mathematics and Computational Science,
               Xiangtan University, Xiangtan 411105, Hunan, China \\
               \and
Chengqing Li \at
              College of Information Engineering,
              Xiangtan University, Xiangtan 411105, Hunan, China \\
              Tel.: +86-731-52639779\\
              Fax: +86-731-58292217\\
              \email{chengqingg@gmail.com}
           \and
Qin Li  \at
               College of Information Engineering,
               Xiangtan University, Xiangtan 411105, Hunan, China \\
\and
Dan Zhang \at
College of Computer Science and Technology, Zhejiang University,
Hangzhou 310027, Zhejiang, China
}

\date{Received: Nov 5 2011}

\maketitle

\begin{abstract}
Recently, a chaotic image encryption algorithm based on perceptron model was proposed.
The present paper analyzes security of the algorithm and finds that the equivalent secret key can
be reconstructed with only one pair of known-plaintext/ciphertext, which is supported by both
mathematical proof and experiment results. In addition, some other security defects are also reported.


\end{abstract}

\section{Introduction}
\label{sec:intro}

The usage of chaos for image encryption has been received intensive attention in the past decade.
This is mainly spurred by the following three points: 1) increasing importance of security of image data as it is
transmitted over all kinds of networks with a more and more higher frequency; 2) low efficiency of the traditional text encryption algorithms like DES
with respect to protecting image data due to some intrinsic features of images such as bulk data capacity, high redundancy,
strong correlation among adjacent pixels, etc; 3) some fundamental features of the chaotic dynamical systems such as ergodicity,
mixing property, sensitivity to initial conditions/system parameter can be considered
analogous to some ideal cryptographic properties such as confusion, diffusion, avalanche properties. According to the
record of \textit{Web of Science}, more than five hundreds of articles on chaos-based image encryption algorithms
have been published in the past fifteen years \cite{YaobinMao:CSF2004}. Some other related papers consider video, audio (speech) or text as encryption
objects \cite{ZQChen:H263:ND10,Jye:SpeechEncrypt:ND11,ChenJY:Joint:TCSII11}.

Roughly speaking, the usage of chaos in designing digital symmetric encryption
algorithms can be categorized as three classes: 1) creating position permutation relations directly or indirectly; 2) generating
pseudo-random bit sequence (PRBS) controlling composition and combination of some basic encryption operations; 3) producing ciphertext
directly by assigning plaintext as initial conditions or control parameters of a chaos system. As native opposite of cryptography, some cryptanalysis work was
also developed and some chaos-based encryption algorithms are found to be not secure of different extents from the viewpoint of
modern cryptology \cite{Ercan:AttackLorenz:TCASII04,AlvarezLi:Rules:IJBC2006,Rhouma:BreakLian:PLA08,Alvarez:BreakingNCA:CNSNS09,LCQ:BreakHCKBA:IJBC11,ZhouJT:Coder:TCSI11}. Due to owning complex dynamical properties, artificial neural network was combined with chaos to design symmetric and public encryption algorithms \cite{Guodh:probabilistic:AI99,Mislovaty:Publiccipher:PRL03,
Zhouts:clippedneural:LNCS04,Leungkc:Chnn:NPL05,LSG:Blockcipher:NC09}. Unfortunately, some of them are found to be equivalent to
much simper form in terms of essential structure and can be broken easily \cite{Li:AttackingCNN2004,Li:AttackingISNN2005}.

In \cite{Wang:perceptron:ND10}, a chaotic image encryption algorithm based on perceptron model, a type of artificial neural network invented in 1957, was
proposed, where two PRBSs, generated from quantized orbit of Lorenz system under given secret key, is used to control the perceptron model to
output cipher-image from the input of plain-image. However, we proved that the seeming complex image encryption algorithm
is equivalent to a stream cipher of exclusive or (XOR) operation. So, it can be easily broken with only one pair of plain-image and
cipher-image. In addition, some other security defects about secret key, insensitive of cipher-image with respect to plain-image and
low randomness of the used PRBS are reported.

The rest of the paper is organized as follows. The next section briefly introduces the
proposed image encryption algorithm. Section~\ref{sec:cryptanalyze} presents detailed cryptanalysis on the encryption algorithm
with experiment results and some other security defects. The last section concludes the paper.

\section{The proposed image encryption algorithm}
\label{sec:algorithm}

The plaintext encrypted by the image encryption algorithm under study is a gray-scale digital image.
Without loss of generality, the plaintext can be represented as a one-dimensional $8$-bit integer sequences $P=\{p_n\}_{n=0}^{N-1}$
by scanning it in a raster order. Correspondingly, the ciphertext is denoted by $P'=\{p'_n\}_{n=0}^{N-1}$. The kernel of the encryption algorithm is based on a
threshold function
\begin{equation*}
f(x)=
\begin{cases}
1, & \text{if } x\ge 0,\\
0, & \text{otherwise,}
\end{cases}
\end{equation*}
which was considered by the proposers as simple variant of a single layer
perceptron model who inputs $m$ variables, $s_0, s_1, \cdots, s_{m-1}$, and outputs $m$ ones by calculating
\begin{equation*}
    g(s_i)=
    \begin{cases}
    1, & \text{if } \left(\sum\limits_{j=0}^{m-1} s_iw_{ij}-\theta_i \right)\geq 0,\\
    0, & \text{otherwise},
    \end{cases}
\end{equation*}
where $w_{ij}$ denotes the weight of the $i$-th input for the $j$-th neuron, $\theta_i$ is the threshold of the $i$-th neuron,
and $i=0\sim m-1$. With these preliminary introduction, the image encryption algorithm under study can be described as follows\footnote{To make the presentation more concise and complete, some notations in the original paper \cite{Wang:perceptron:ND10} are modified, and some details about the algorithm are also amended under the condition that its security property is not influenced.}.

\begin{itemize}
\item \textit{The secret key}: initial state $(x_0^*, y_0^*, z_0^*)$ and the step length $h$ of an approximation method solving the Lorenz system
\begin{equation}
\left\{
\begin{array}{lcl}
$\.{x}$ &=& ay-ax, \\
$\.{y}$ &=& cx-xz-y, \\
$\.{z}$ &=& xy-bz,
\end{array}
\right.
\label{eq:lorenz}
\end{equation}
under fixed control parameters $(a, b, c)=(10, \frac{8}{3}, 28)$.

\item \textit{The initialization procedures}:

\par 1) in double-precision floating-point arithmetic, solve the Lorenz system (\ref{eq:lorenz}) with the fourth-order Runge-Kutta method of step $h$ iteratively
3001 times from $(x_0^*, y_0^*, z_0^*)$ and obtain the current approximation state $(x_0, y_0, z_0)$ .

\par 2) run the above solution approximation step seven more times from the current approximation state to get $\{(x_j, y_j, z_j)\}_{j=1}^7$, and set
\begin{equation*}
    w_{j} =
    \begin{cases}
    1, & \text{if } (x_j-\check{x})/(\hat{x}-\check{x}) \geq 0.5,\\
    -1, & \text{otherwise},
    \end{cases}
    \label{eq:weight1}
\end{equation*}
and
\begin{equation*}
   \widetilde{w}_{j} =
    \begin{cases}
    1, & \text{if } (y_j-\check{y})/(\hat{y}-\check{y}) \geq 0.5,\\
    -1, & \text{otherwise},
    \end{cases}
    \label{eq:weight2}
\end{equation*}
for $j=0\sim 7$, where $\hat{x}=\max(\{x_j\}_{j=0}^7)$, $\check{x}=\min(\{x_j\}_{j=0}^7)$, $\hat{y}=\max(\{y_j\}_{j=0}^7)$, $\check{y}=\min(\{y_j\}_{j=0}^7)$.

\par 3) reset the current approximation state of the Lorenz
system as
\begin{equation*}
\left\{
\begin{array}{lcl}
x_0 & = & \check{x}+(x_8-\check{x})\left(\sum\limits_{j=0}^{7}(w_j+1) \cdot 2^{j-1} \oplus r\right)/256, \\
y_0 & = & \check{y}+(y_8-\check{y})\left(\sum\limits_{j=0}^{7}(\widetilde{w}_j+1) \cdot 2^{j-1} \oplus r\right)/256, \\
z_0 & = & z_8,
\end{array}
\right.
\end{equation*}
where $r=\lfloor(z_8- \lfloor z_8 \rfloor)\cdot 256\rfloor$.

\par 4) repeat the above two steps $N-1$ times and get two sequences
$\{w_k\}_{k=0}^{8N-1}$ and $\{\widetilde{w}_k\}_{k=0}^{8N-1}$.

\item \textit{The encryption procedure}:
for the $n$-th plain-byte $p_n=\sum_{i=0}^7 p_{n,i}\cdot 2^i$, obtain the
corresponding cipher-byte $p'_n=\sum_{i=0}^7 p'_{n,i}\cdot 2^i$ by
calculating
\begin{equation}
    p'_{n,i} =
    \begin{cases}
    f(p_{n,i} w_k + c_k \widetilde{w}_k- \theta_k), & \text{if } w_k =1,\\
    f(p_{n,i} w_k - c_k \widetilde{w}_k+ \theta_k), & \text{otherwise},
    \end{cases}
\label{eq:encryption}
\end{equation}
where $c_k=-w_k/2$,
$\theta_k =((w_k+1)/2)\oplus ((\widetilde{w}_k+1)/2)$ and $k=8\cdot n+i$.

\item \textit{The decryption procedure} is the same as the encryption one except that the locations of $p_{n,i}$ and $p'_{n,i}$ in the encryption function
(\ref{eq:encryption}) are swapped.
\end{itemize}

\section{Cryptanalysis}
\label{sec:cryptanalyze}

\subsection{Known-plaintext attack}

It is well-known that any detail of an encryption algorithm, except the secret key, should be public. This is called
Kerckhoffs's principle, which was reformulated by Claude Shannon as Shannon's maxim, ``The enemy knows the algorithm."
Under this principle, the known-plaintext attack is a cryptanalysis model
where the attacker can access some samples of both the plaintext, and the corresponding
ciphertext (encrypted version). The objective of the model is to reveal some (even all) information about the secret key,
which is then used to decrypt other ciphertext encrypted with the same secret key. Note that feasibility of the known-plaintext attack
is based on repeated usage of secret key, namely impracticability of one-time pad, which is caused by the following three
problems: 1) impossibility of software source of perfectly random bits; 2) secure generation and exchange of the one-time pad material
of not shorter length than that of the plaintext; 3) complex secret key management preventing the secret key is reused in whole or part.

Strength of any encryption algorithm against the known-plaintext attack is one of the most important factors evaluating
its security. Unfortunately, we found that the image encryption algorithm under study can be broken with only one pair of
plaintext and the corresponding ciphertext.

\begin{theorem}
For $n=0\sim N-1$ and $i=0\sim 7$,
\begin{equation}
p'_{n,i} =
\begin{cases}
p_{n,i},            &   \mbox{if } w_k =1,\\
\overline{p_{n,i}}, &   \mbox{otherwise},
\end{cases}
\label{eq:simplify}
\end{equation}
where $k=8n+i$, $\overline{x}=(x\oplus 1)$.
\label{theorem:kpa}
\end{theorem}
\begin{proof}
To proof the theorem, we study the four possible combination of $(w_k, p_{n,i})$ as follows.
\begin{itemize}
\item $(w_k, p_{n,i})=(1, 1)$:
\begin{eqnarray*}
p'_{n,i} & = & f( p_{n,i}\cdot w_k+ c_k\cdot \widetilde{w}_k- \theta_k) \\
         & = & f( 1 + (-1/2) \cdot \widetilde{w}_k- \theta_k) \\
         & = &
         \begin{cases}
         f( 1 -1/2- 0),   & \text{if } \widetilde{w}_k =1,\\
         f( 1 +1/2-1),   & \text{otherwise}
         \end{cases}  \\
         & = & f(1/2) \\
         & = & 1.
\end{eqnarray*}

\item $(w_k, p_{n,i})=(1, 0)$:
\begin{eqnarray*}
p'_{n,i} &=& f( p_{n,i}\cdot w_k+ c_k\cdot \widetilde{w}_k- \theta_k) \\
         &=& f( 0 + (-1/2)\cdot \widetilde{w}_k- \theta_k) \\
         &=&
            \begin{cases}
            f( 0 -1/2-0), & \text{if } \widetilde{w}_k =1,\\
            f( 0 +1/2-1), & \text{otherwise}
            \end{cases} \\
            &=& f(-1/2) \\
            &=& 0.
\end{eqnarray*}

\item $(w_k, p_{n,i})=(-1, 1)$:
\begin{eqnarray*}
p'_{n,i} &=& f(p_{n,i}\cdot w_k - c_k\cdot \widetilde{w}_k+ \theta_k) \\
         &=& f( -1 - 1/2\cdot \widetilde{w}_k+ \theta_k) \\
         &=&   \begin{cases}
         f( -1 -1/2 +1), & \text{if } \widetilde{w}_k =1,\\
         f( -1 +1/2 +0), & \text{otherwise}
         \end{cases} \\
         &=& f(-1/2) \\
         &=& 0.
\end{eqnarray*}

\item $(w_k, p_{n,i})=(-1, 0)$:
\begin{eqnarray*}
    p'_{n,i} &=& f(p_{n,i}\cdot w_k - c_k\cdot \widetilde{w}_k+ \theta_k) \\
            &=& f( 0 - 1/2\cdot \widetilde{w}_k+ \theta_k) \\
            &=&
            \begin{cases}
            f( 0 -1/2 +1), & \text{if } \widetilde{w}_k =1,\\
            f( 0 +1/2 +0), & \text{otherwise}
            \end{cases}\\
            &=& f(1/2) \\
            &=& 1.
\end{eqnarray*}
\end{itemize}

Combining the above four cases, one has
\begin{equation*}
p'_{n,i} =
    \begin{cases}
    1, & \text{if } p_{n,i}=1 \text{ and } w_k =1,\\
    0, & \text{if } p_{n,i}=0 \text{ and } w_k =1,\\
    0, & \text{if } p_{n,i}=1 \text{ and } w_k =-1,\\
    1, & \text{if } p_{n,i}=0 \text{ and } w_k =-1,
    \end{cases}
\end{equation*}
which means
\begin{equation*}
p'_{n,i} =
\begin{cases}
p_{n,i},             & \text{if } w_k =1,\\
\overline{p_{n,i}},  & \text{otherwise},
\end{cases}
\end{equation*}
Thus, the theorem is proved.\qed
\end{proof}

Observe Theorem~\ref{theorem:kpa}, one has
\begin{equation}
p'_{n,i} = p_{n,i} \oplus \overline{w'_k},
\label{eq:moresimpler}
\end{equation}
where $w'_k=(w_k+1)/2$. From Eq.~(\ref{eq:moresimpler}), one can get
\begin{equation*}
(p_n\oplus p'_n)=\eta_n,
\end{equation*}
where $\eta_n=\sum_{i=0}^7 \overline{w'_{8n+i}}\cdot 2^i$
and $n=0\sim N-1$. For any another cipher-image $Q'=\{q'_n\}_{n=0}^{N-1}$ encrypted with the same secret key, one can easily reveal the corresponding
plain-image $Q=\{q_n\}_{n=0}^{N-1}$ by calculating
\begin{equation*}
q_n= q'_n\oplus \eta_n
\end{equation*}
for $n=0\sim N-1$, which means
$H=\{\eta_n\}_{n=0}^{N-1}$ can work as equivalent secret key.

To verify performance of the above attack, some experiments were made. Figure~\ref{figure:knownplaintext}
shows a plain-image ``Lenna" of size $512 \times 512$ and the encryption result with the secret key $(x_0^*, y_0^*, z_0^*)=(1, 1, 0)$, and $h=10^{-1}$. A mask image $H$ is constructed by XORing datum of Fig.~\ref{figure:knownplaintext}a) and Fig.~\ref{figure:knownplaintext}b)
byte by byte, and shown in Fig.~\ref{figure:attack}a). The mask image is then used to decrypt another cipher-image shown in Fig.~\ref{figure:attack}b)
and the result is shown in Fig.~\ref{figure:attack}c), which is identical with the original plain-image ``Baboon".

\begin{figure}[!htb]
\centering
\begin{minipage}{\figwidth}
\centering
\includegraphics[width=\textwidth]{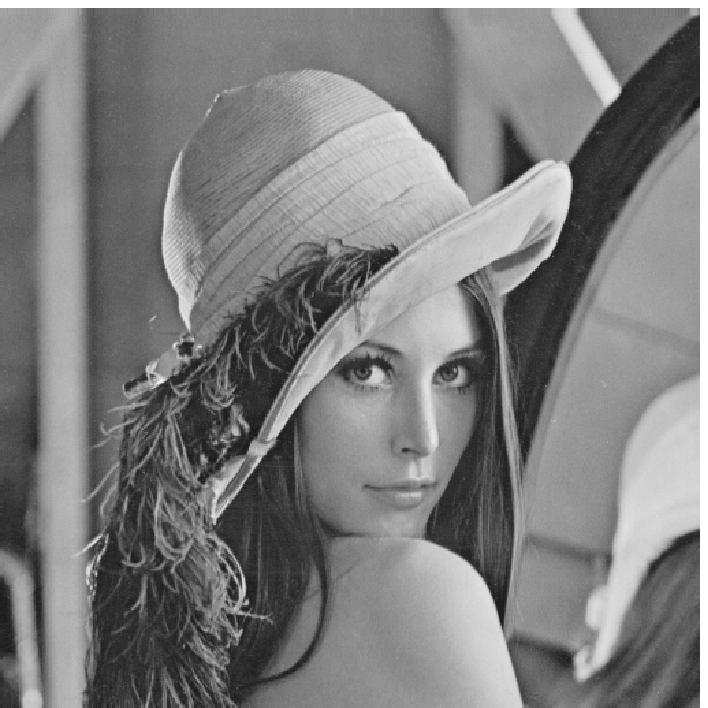}
a)
\end{minipage}
\begin{minipage}{\figwidth}
\centering
\includegraphics[width=\textwidth]{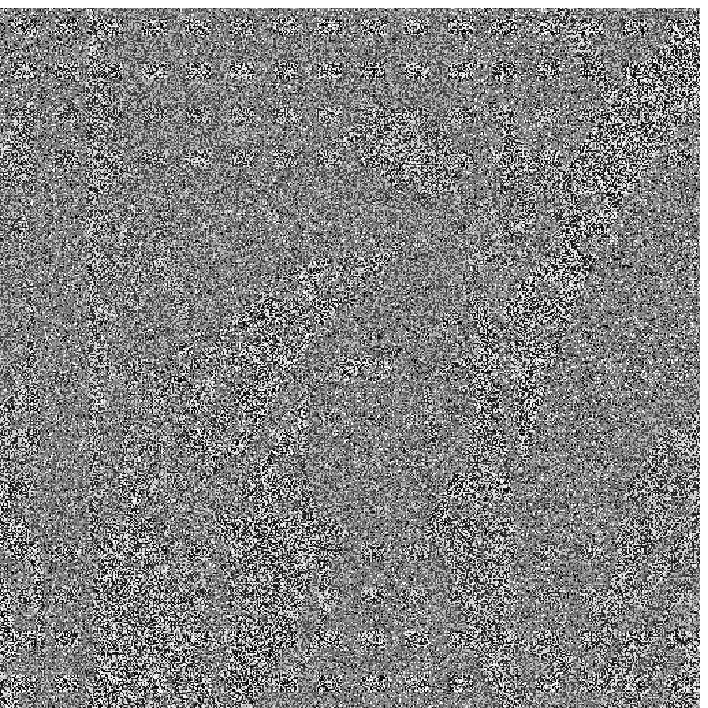}
b)
\end{minipage}
\caption{One known plain-image and the corresponding cipher-image: a) known plain-image ``Lenna", b) cipher-image of Fig.~\ref{figure:knownplaintext}a).}
\label{figure:knownplaintext}
\end{figure}

\begin{figure}[!htb]
\centering
\begin{minipage}{\figwidth}
\centering
\includegraphics[width=\textwidth]{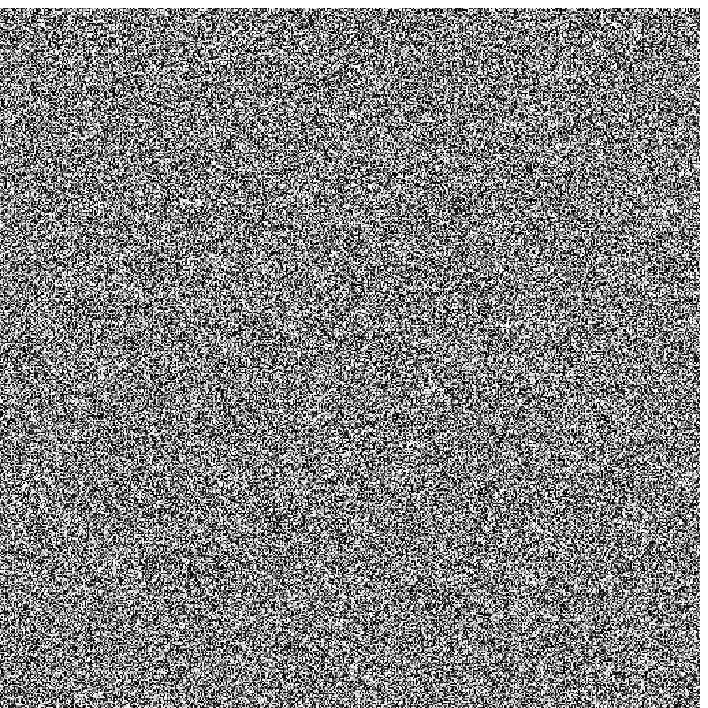}
a)
\end{minipage}\\
\begin{minipage}{\figwidth}
\centering
\includegraphics[width=\textwidth]{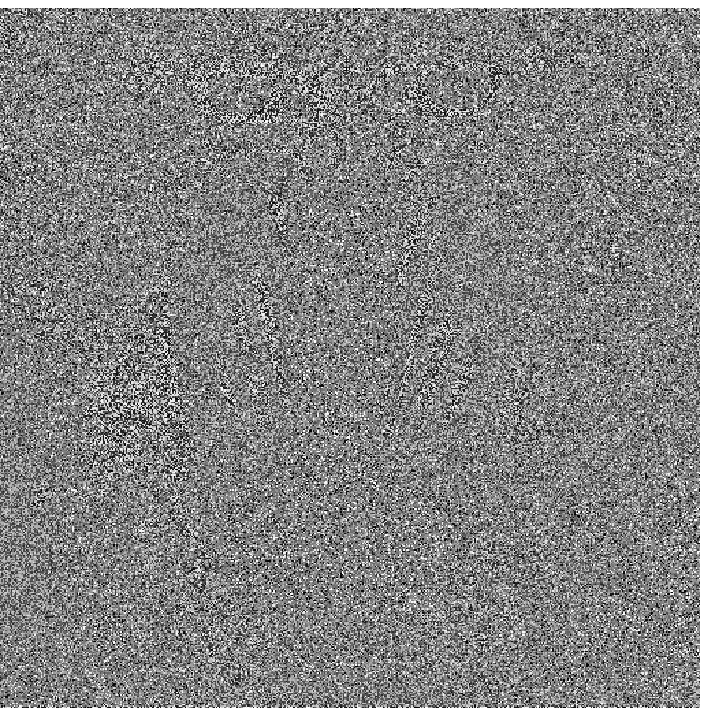}
b)
\end{minipage}
\begin{minipage}{\figwidth}
\centering
\includegraphics[width=\textwidth]{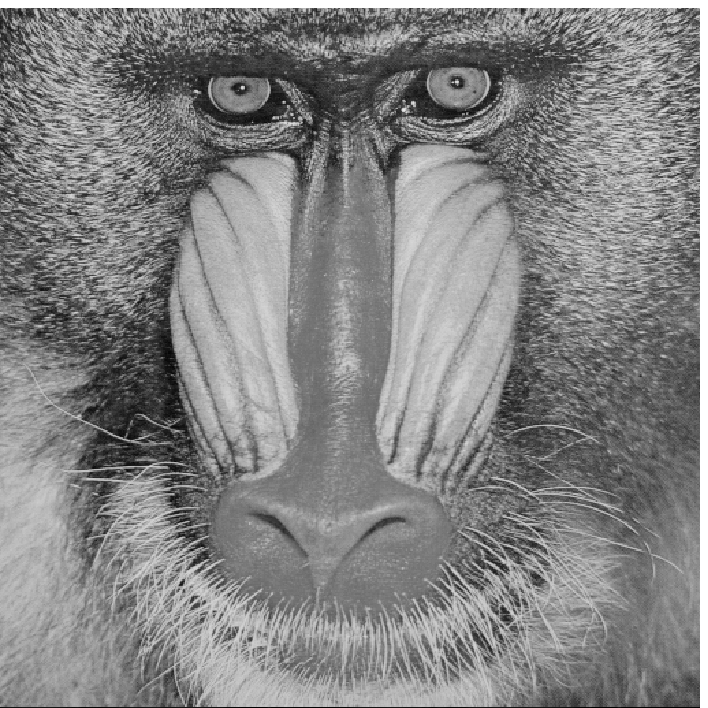}
c)
\end{minipage}
\caption{Known-plaintext attack:
a) the mask image $H$,
b) cipher-image of plain-image ``Baboon",
c) the recovered image of Fig.~\ref{figure:attack}b).}
\label{figure:attack}
\end{figure}

\subsection{Some other security defects}

In this subsection, we further report some other security defects of the image encryption algorithm under study.

\begin{itemize}
\item \textit{Low sensitivity with respect to change of plain-image}

If encryption results of an encryption algorithm are not sensitive to change of plaintext in a
significant degree, the attacker can make predictions about the plaintext from the given ciphertext only.
The desirable property of encryption algorithm is called avalanche effect in field of cryptography.
The property is even more important for image encryption algorithms since image datum and its watermarked
versions, generally slight modified versions of the original ones, are encrypted often at the same time and stored
at the same place. The avalanche effect is quantitatively measured by how every bit of ciphertext is changed when only one bit
of plaintext is modified. As for the image encryption algorithm under study, change of one bit of plain-image
can only influent one bit of the same location in the corresponding ciphertext, which disobeys expected property
of a secure encryption algorithm very far.

\item \textit{Insufficient randomness of the two used PRBSs}

The complex dynamical properties of chaos systems demonstrated in continuous domain
make they are considered as good generation source of PRBS. However, chaos systems can not
do the expected things in general since their dynamical properties will definitely degenerate in
digital domain, where everything is stored and calculated in limited
precision \cite{Li:DPWLCM:IJBC2005}. In addition, some chaos systems need numerical approximation, the used simulation methods and related parameters will
have different influences on degradation of the dynamical properties of the chaos system also. As for the image encryption algorithm
under study, trajectory of Lorenz system is continuous, which means that any two consecutive
simulated states are always correlated in a high degree. As a consequence, the bits derived from neighboring
states will also be correlated closely. Furthermore, the smaller the step length $h$ is, the stronger the
correlation will be \cite{Li:AttackingBitshiftXOR2007}. Meanwhile, the value of $h$ should be small enough since the accumulated error
of the fourth-order Runge-Kutta method is $O(h^4)$. So, step length $h$ should not be used as a sub-key since its valid
scope can not be defined clearly. To verify this point, we employed the NIST statistical
test suite \cite{Rukhin:TestPRNG:NIST} to test the randomness of 100 binary sequences of length $512\times 512 \times 8= 2,097,152$ bits (the number of bits in $\{\overline{w'_k}\}$ used for the encryption of $512\times 512$ gray-scale images). Note that the 100 binary sequences are generated by
randomly selected initial conditions $(x_0^*, y_0^*, z_0^*)$ and a given step length $h$. For each test, the default significance level 0.01 was adopted.
The passing rates in terms of binary matrix rank test, calculating the rank of disjoint sub-matrices of the entire sequence,
under different values of step length $h$ are shown in Fig.~\ref{figure:rank}, which agree with
the estimation. We found that the PRBSs generated by the method used in the image encryption algorithm under study
can not pass most tests in the test suite even for the step length $h$ under which the rank test can be passed.
When $h=0.1$, the test results are shown in Table~\ref{table:test}, from which one can see that the
used pseudo-random bit generator is not a good source of PRBS.

\begin{figure}[!htb]
\begin{minipage}{\figwidth}
\includegraphics[width=1.8\textwidth]{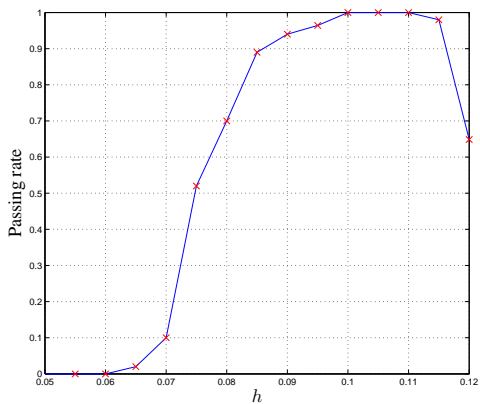}
\end{minipage}
\caption{Passing rate of the rank test under different values of step length $h$.}
\label{figure:rank}
\end{figure}

\begin{table}[!htbp]
\centering\caption{The performed tests with respect to a
significance level 0.01 and the number of sequences passing each
test in 100 randomly generated sequences.}
\begin{tabular}{p{3.85cm}|c}
\hline
{Name of Test}     & {Number of Passed Sequences}               \\
\hline\hline Frequency                                      & 93\\
\hline Block Frequency ($m=128$)                            & 100\\
\hline Cumulative Sums-Forward                              & 94\\
\hline Runs                                                 & 0\\
\hline Rank                                                 & 100\\
\hline Serial ($m=16$)                                      & 0\\
\hline Spectral Test                                        & 0\\
\hline Random Excursions(x=1)                               & 0\\
\hline Approximate Entropy ($m=10$)                         & 0\\
\hline Longest Runs of Ones (m=10000)                       & 0\\
\hline Non-overlapping Template ($m=9$, $B=000000001$)      & 0\\
\hline
\end{tabular}
\label{table:test}
\end{table}
\end{itemize}

\section{Conclusion}

In this paper, the security of a chaotic image encryption algorithm based on perceptron model has been investigated in detail.
The seeming complex encryption algorithm was proved to be equivalent to a stream cipher essentially, and so it can be broken with
only one pair of known-plaintext/ciphretext. In addition, some other security defects of the encryption algorithm, including
insensitivity with respect to change of plain-image and low randomness of PRBS used are also reported. In all,
this cryptanalysis paper shows us again that security of an encryption algorithm should be mainly built on good essential
structure of the whole encryption operations instead of so-called complex theories.

\begin{acknowledgements}
This research was supported by the National Natural Science Foundation of China (No. 61100216), Scientific Research Fund of Hunan Provincial Education Department (No. 11B124), and Ningbo Natural Science Foundation (No. 2011A610194).
\end{acknowledgements}

\bibliographystyle{spmpsci}
\bibliography{Perceptron}
\end{document}